\documentclass[amscd,amssymb,reqno
]{amsart}
\usepackage{amsmath,amsthm,amscd,amssymb}
\usepackage{latexsym}
\usepackage{color}
\usepackage{mathrsfs}
\usepackage{geometry}
\usepackage{graphicx}

\usepackage{relsize}
\usepackage{tikz}
\usetikzlibrary{calc}

\usepackage{subfig}

\usepackage[all]{xy}

\numberwithin{equation}{section}

\newtheorem{Thm}{Theorem}[section]

\newtheorem{Lem}[Thm]{Lemma}
\newtheorem{Prop}[Thm]{Proposition}

\newtheorem{Rmk}[Thm]{Remark}

\newtheorem{Fact}[Thm]{Fact}
\newtheorem{Notation}[Thm]{Notation}

\renewcommand{\l}{\lambda}
\renewcommand{\o}{\omega}
 
\newcommand{\N}{\mathbb{N}}

\newcommand{\R}{\mathbb{R}}
\newcommand{\E}{\mathbb{E}}

\renewcommand{\L}{\Lambda}
\renewcommand{\E}{\mathbb{E}_{(\aaa),\L}}

\newcommand{\M}{\mathcal{M}_l}
\newcommand{\V}{|\Lambda|}
\newcommand{\aaa}{\mu, \vec{\lambda}, \beta}
\newcommand{\m}{\rho, \vec{u}, E}
\renewcommand{\M}{N, \vec{P}, H}

\usepackage{etoolbox}

\makeatletter
\let\old@tocline\@tocline
\let\section@tocline\@tocline
\newcommand{\subsection@dotsep}{4.5}
\newcommand{\subsubsection@dotsep}{4.5}
\patchcmd{\@tocline}
  {\hfil}
  {\nobreak
     \leaders\hbox{$\m@th
        \mkern \subsection@dotsep mu\hbox{.}\mkern \subsection@dotsep mu$}\hfill
     \nobreak}{}{}
\let\subsection@tocline\@tocline
\let\@tocline\old@tocline

\patchcmd{\@tocline}
  {\hfil}
  {\nobreak
     \leaders\hbox{$\m@th
        \mkern \subsubsection@dotsep mu\hbox{.}\mkern \subsubsection@dotsep mu$}\hfill
     \nobreak}{}{}
\let\subsubsection@tocline\@tocline
\let\@tocline\old@tocline

\let\old@l@subsection\l@subsection
\let\old@l@subsubsection\l@subsubsection

\def\@tocwriteb#1#2#3{%
  \begingroup
    \@xp\def\csname #2@tocline\endcsname##1##2##3##4##5##6{%
      \ifnum##1>\c@tocdepth
      \else \sbox\z@{##5\let\indentlabel\@tochangmeasure##6}\fi}%
    \csname l@#2\endcsname{#1{\csname#2name\endcsname}{\@secnumber}{}}%
  \endgroup
  \addcontentsline{toc}{#2}%
    {\protect#1{\csname#2name\endcsname}{\@secnumber}{#3}}}%

\newlength{\@tocsectionindent}
\newlength{\@tocsubsectionindent}
\newlength{\@tocsubsubsectionindent}
\newlength{\@tocsectionnumwidth}
\newlength{\@tocsubsectionnumwidth}
\newlength{\@tocsubsubsectionnumwidth}
\newcommand{\settocsectionnumwidth}[1]{\setlength{\@tocsectionnumwidth}{#1}}
\newcommand{\settocsubsectionnumwidth}[1]{\setlength{\@tocsubsectionnumwidth}{#1}}
\newcommand{\settocsubsubsectionnumwidth}[1]{\setlength{\@tocsubsubsectionnumwidth}{#1}}
\newcommand{\settocsectionindent}[1]{\setlength{\@tocsectionindent}{#1}}
\newcommand{\settocsubsectionindent}[1]{\setlength{\@tocsubsectionindent}{#1}}
\newcommand{\settocsubsubsectionindent}[1]{\setlength{\@tocsubsubsectionindent}{#1}}

\renewcommand{\l@section}{\section@tocline{1}{\@tocsectionvskip}{\@tocsectionindent}{}{\@tocsectionformat}}%
\renewcommand{\l@subsection}{\subsection@tocline{1}{\@tocsubsectionvskip}{\@tocsubsectionindent}{}{\@tocsubsectionformat}}%
\renewcommand{\l@subsubsection}{\subsubsection@tocline{1}{\@tocsubsubsectionvskip}{\@tocsubsubsectionindent}{}{\@tocsubsubsectionformat}}%
\newcommand{\@tocsectionformat}{}
\newcommand{\@tocsubsectionformat}{}
\newcommand{\@tocsubsubsectionformat}{}
\expandafter\def\csname toc@1format\endcsname{\@tocsectionformat}
\expandafter\def\csname toc@2format\endcsname{\@tocsubsectionformat}
\expandafter\def\csname toc@3format\endcsname{\@tocsubsubsectionformat}
\newcommand{\settocsectionformat}[1]{\renewcommand{\@tocsectionformat}{#1}}
\newcommand{\settocsubsectionformat}[1]{\renewcommand{\@tocsubsectionformat}{#1}}
\newcommand{\settocsubsubsectionformat}[1]{\renewcommand{\@tocsubsubsectionformat}{#1}}
\newlength{\@tocsectionvskip}
\newcommand{\settocsectionvskip}[1]{\setlength{\@tocsectionvskip}{#1}}
\newlength{\@tocsubsectionvskip}
\newcommand{\settocsubsectionvskip}[1]{\setlength{\@tocsubsectionvskip}{#1}}
\newlength{\@tocsubsubsectionvskip}
\newcommand{\settocsubsubsectionvskip}[1]{\setlength{\@tocsubsubsectionvskip}{#1}}

\patchcmd{\tocsection}{\indentlabel}{\makebox[\@tocsectionnumwidth][l]}{}{}
\patchcmd{\tocsubsection}{\indentlabel}{\makebox[\@tocsubsectionnumwidth][l]}{}{}
\patchcmd{\tocsubsubsection}{\indentlabel}{\makebox[\@tocsubsubsectionnumwidth][l]}{}{}

\newcommand{\@sectypepnumformat}{}
\renewcommand{\contentsline}[1]{%
  \expandafter\let\expandafter\@sectypepnumformat\csname @toc#1pnumformat\endcsname%
  \csname l@#1\endcsname}
\newcommand{\@tocsectionpnumformat}{}
\newcommand{\@tocsubsectionpnumformat}{}
\newcommand{\@tocsubsubsectionpnumformat}{}
\newcommand{\setsectionpnumformat}[1]{\renewcommand{\@tocsectionpnumformat}{#1}}
\newcommand{\setsubsectionpnumformat}[1]{\renewcommand{\@tocsubsectionpnumformat}{#1}}
\newcommand{\setsubsubsectionpnumformat}[1]{\renewcommand{\@tocsubsubsectionpnumformat}{#1}}
\renewcommand{\@tocpagenum}[1]{%
  \hfill {\mdseries\@sectypepnumformat #1}}

\let\oldappendix\appendix
\renewcommand{\appendix}{%
  \leavevmode\oldappendix%
  \addtocontents{toc}{%
    \protect\settowidth{\protect\@tocsectionnumwidth}{\protect\@tocsectionformat\sectionname\space}%
    \protect\addtolength{\protect\@tocsectionnumwidth}{2em}}%
}
\makeatother

\makeatletter
\settocsectionnumwidth{2em}
\settocsubsectionnumwidth{2.5em}
\settocsubsubsectionnumwidth{3em}
\settocsectionindent{1pc}%
\settocsubsectionindent{\dimexpr\@tocsectionindent+\@tocsectionnumwidth}%
\settocsubsubsectionindent{\dimexpr\@tocsubsectionindent+\@tocsubsectionnumwidth}%
\makeatother

\settocsectionvskip{10pt}
\settocsubsectionvskip{0pt}
\settocsubsubsectionvskip{0pt}

\let\oldtableofcontents\tableofcontents
\renewcommand{\tableofcontents}{%
  \vspace*{-\linespacing}
  \oldtableofcontents}

\date{\today}

\title[Entropy minimization and convergence]
{On entropy minimization and convergence} 
\author{S. Dostoglou}
\address{Department of Mathematics, 
          University of Missouri, 
          Columbia, MO 65211}
\email{dostoglous@missouri.edu}
\author{A. Hughes}
\address{Department of Mathematics, 
          University of Missouri, 
          Columbia, MO 65211}
\email{amhf44@mail.missouri.edu}

\author{Jianfei Xue}
\address{Department of Mathematics, 
   University of Arizona, 
          Tucson, AZ 85721}
\email{jxue@math.arizona.edu}

\thanks{Work supported by NSF Grant No.~DMS-1440140 and NSA Grant No.~H98230-18-1-0269  while the first two authors were in residence at the Mathematical Sciences Research Institute in Berkeley, California during the Fall 2018 semester.}

\begin{document}




\begin{abstract}
We examine the minimization of information entropy for measures on the phase space of bounded domains, subject to constraints that are averages of grand canonical distributions. We describe the set of all such constraints and show that it equals the set of averages of all  probability measures absolutely continuous with respect to the standard measure on the phase space (with the exception of the measure concentrated on the empty configuration). We also investigate how the set of constrains relates to the domain of the microcanonical thermodynamic limit entropy. We then show that, for fixed constraints, the parameters of the corresponding grand canonical distribution converge, as volume increases, to the corresponding parameters (derivatives, when they exist) of the thermodynamic limit entropy.

The results hold when the energy is the sum of any stable, tempered  
interaction potential that satisfies the Gibbs variational principle (e.g.~Lennard-Jones) and the kinetic energy. 

The same tools and the strict convexity of the thermodynamic limit pressure for continuous systems (valid whenever the Gibbs variational principle holds) give solid foundation to the folklore local homeomorphism between thermodynamic and macroscopic quantities.
\end{abstract}

\maketitle

\tableofcontents

\section{Introduction}

For systems with a large number of components, information entropy optimization subject to a few given constraints is widely used in several fields, see \cite{BKM} for example.
In statistical mechanics, for bounded domains in space, one minimizes the information entropy of measures, absolutely continuous with respect to the standard measure on the phase space, subject to appropriate macroscopic constraints for density, velocity, and energy, see \cite{J}, \cite{Z}. The minimizing measure, a grand-canonical Gibbs distribution, is characterized by parameters conjugate to the constraints, in the sense that differentiating the logarithm of the partition function of the measure with respect to these parameters returns the constraints. 
On the other hand, it is standard how, as the volume increases, one defines microcanonical entropy in terms of thermodynamic limits as a function of the given  macroscopic quantities. Convex conjugation  
(or derivatives, when they exist) now define the thermodynamic parameters at the limit.

Motivated by:
\begin{itemize}
\item 
the need to know in advance in numerical simulations  the range of constraints that can be used when minimizing information entropy on bounded domains, as well as
\item
the comparison of hydrodynamic equations from statistical mechanics on bounded domains, which heavily relies on minimizing information entropy (see \cite{BAR} and \cite{Z}), to the standard hydrodynamic limit approach (see \cite{OVY}), 
\end{itemize}
we ask the following:
\begin{itemize}
\item 
what are the appropriate constraints for a given domain and how does the set of these constraints change with domain?
\item
do the thermodynamic parameters for finite volume converge to the thermodynamic limit parameters while the macroscopic constraints remain fixed?
\end{itemize} 
{
In this article we address these questions. In doing so we also obtain some other results which, to our knowledge, are new to the literature. 

To start, in section \ref{functions}, we examine rigorously the constrained  information entropy minimization over a bounded domain. It is standard that if we assume the constrains are averages of some grand canonical distribution this grand canonical distribution is the unique minimizer, cf.~\cite{Gibbs}. We give an explicit description of the set of possible averages of grand-canonical distributions on bounded domains, cf.~Theorem \ref{bijection}. Furthermore, Theorem \ref{prop: in the interior}  asserts that this set in fact coincides with the averages of all measures that are absolutely continuous with respect to the standard measure on phase space (with the exception of a trivial case: the measure concentrated on the empty configuration). This shows that the usual assumption on the constrains is not restrictive. It is worth mentioning that the usual formal Lagrange multipliers argument (cf.~\cite{Z}) is not employed here. In fact, we find the application of Lagrange multipliers in this setting questionable, cf.~section \ref{constrained mini on L}.

We then proceed to thermodynamic limits in section \ref{Limits}. After recalling definitions and standard properties of the thermodynamic limit microcanonical entropy and pressure we discuss the strict convexity property of the thermodynamic limit pressure. The arguments presented here rely on the Gibbs variational principle. Notice here that, by the convex conjugacy between pressure and microcanonical entropy, the strict convexity of pressure gives differentiability of entropy which plays a role in sections \ref{Main} and \ref{localbijection}.

In section \ref{Main} we answer the question of convergence of thermodynamic parameters. We start by showing in Theorem \ref{interiors} that the domain of the thermodynamic limit entropy (where it is not $-\infty$) is nonempty and also give it an explicit characterization which improves the results in \cite{L} and \cite{M-L}. 
Then Theorem \ref{MainTheorem} shows that, whenever we fix macroscopic constraints in the thermodynamic limit entropy, the finite volume parameters of the corresponding grand-canonical distributions converge, as the domain tends to infinity, to the thermodynamic limit parameters corresponding to the same macroscopic quantities.

The main point here is that when the domain is bounded information entropy and the log-partition function are related via convex conjugation (Legendre-Fenchel transform). But at the thermodynamic limit the convex conjugate of pressure is microcanonical entropy. As pressure is the limit of log-partition functions, to relate information entropy (for macroscopic constraints  that do not change as the domain increases) to microcanonical entropy (at the same macroscopic values), we need to know convergence of the conjugates to the conjugate of the limit. We show that the standard results for this type of convergence from convex analysis, involving convergence of epigraphs of the functions in question, apply. In the presence of derivatives at the limit, already established in work by Georgii \cite{Gii} for example, the results follow.

Related work on determining parameters of Gibbs states via maximum likelihood methods on bounded domains appears in \cite{CG}, \cite{Gi}, \cite{DL}. We compare this to our results in section \ref{MaxLike}.   
Note that convergence of {\em microcanonical entropy} on bounded domains to the thermodynamic limit microcanonical entropy is known, see \cite{RP}. In addition, an argument that appears in \cite{L} and \cite{M-L} shows that the information entropy of the {\em canonical} Gibbs distributions for a {\em fixed} inverse temperature converges to the microcanonical thermodynamic entropy evaluated at the energy corresponding to that fixed temperature. We compare this convergence to our results in section \ref{LanAgain}.

Finally, section \ref{localbijection} uses strict convexity results from section \ref{Limits} to show that there is indeed a local homeomorphism between macroscopic and thermodynamic limit parameters. The existence of such a homeomorphism is often used, and is in fact a cornerstone of some seminal work 
\cite[p.~530, p.~556]{OVY}. For the $1$-dimensional case see \cite[Proposition 5.3]{LR}, \cite[Theorem 10.2]{V}. \\
}

\noindent
{\bf Sign Conventions:} We do not use a negative sign for the coefficient of the energy in the exponent of the grand canonical distribution and we define information entropy to be the average of the logarithm of the distribution function. In this way the natural domain for the grand canonical distribution is on negative inverse temperatures and the information entropy on finite volumes here is that of \cite{K}, but the opposite of \cite{Z}. These conventions avoid the awkward minus sign in front of the energy expectation and render the conjugate of information entropy as the log-partition function. But, since we work with convex functions on finite volumes, at the limit we get the opposite of the concave microcanonical entropy.

We present our results for space dimension $3$, but it will be clear that they hold in any dimension.

\section{Minimum Entropy on Bounded Domains}\label{functions}

{It is standard that a finite grand-canonical Gibbs distribution is the unique minimizer of information (Gibbs) entropy among all distributions with the same average as the grand canonical, see \cite[p.~130]{Gibbs}. The main point of this section is to describe the set of all such averages in Theorem \ref{bijection} on bounded domains and how this set depends on the domain. Theorem \ref{prop: in the interior} shows that this set includes the averages of all measures absolutely continuous with respect to the standard measure on phase space (with the exception of the measure concentrated on the empty configuration). The rest of the section collects standard results on information entropy, its minimizations, and the log-partition function, with proofs where these are not available in the references.}

\subsection{Measures and Notation} 
\label{basics}
{For  $\L\subset \R^3$ bounded and measurable set with volume $|\L|>0$,}
we work on the phase space
\begin{equation}
\begin{split}
       \mathfrak{X}_\L = \bigsqcup_{n\geq 0} \left( \Lambda^n \times \R^{3n}\right)= \{\emptyset\} \sqcup \left( \L\times \R^3 \right) \sqcup \left(\L^2 \times \R^6 \right) \sqcup \ldots 
\end{split}
\end{equation}
where $\{\emptyset\}$ indicates the empty configuration consisting of no particles.
Unless otherwise specified, we assume the underlying $\sigma$-algebra to be $\bigcup_{n\geq 0} \mathcal B(\L^n\times \R^{3n})$ where $\mathcal B(\L^n\times \R^{3n})$ is the Borel $\sigma$-algebra on $\L^n\times \R^{3n}$.
{In general $q$ and $p$ stand for the position and velocity of a particle, respectively.} We take the reference measure on $\mathfrak{X}_\L$ to be 
\begin{equation}
\begin{split}
       \omega=
        \sum_{n \geq 0} 
        \omega_n\left(dq_1, \ldots, dq_n,dp_1, \ldots, dp_n\right)
\end{split}
\end{equation}
where
\begin{equation}
\begin{split}
        \o_0(\{\emptyset\}) = 1,\quad        
        \omega_n(dq_1, \ldots, dq_n,dp_1, \ldots, dp_n) 
        =  
        \frac{1}{n!}
        \ dq_1 \ldots dq_n \ 
        dp_1 \ldots dp_n,\quad n \geq 1.
\end{split}
\end{equation}
Then
\begin{equation}
\begin{split}
       \int_ {\mathfrak{X}_\L}  
       f  \ d\omega 
       = 
       \sum_{n \geq 0}
       \int_{\L^n \times \R^{3n}}f(q_1,\ldots,q_n,p_1,\ldots,p_n) 
       \ d\omega_n.
\end{split}
\end{equation}
Throughout, $\mathcal P_\L$ denotes the space of probability measures on $\mathfrak X_\L$.

For any configuration $(\tilde q,\tilde p)\in \mathfrak X_\L$, the total energy $H(\tilde q,\tilde p)$ consists {of} both kinetic energy and potential energy:
\begin{equation} \label{energy}
\begin{split}
       H: \mathfrak{X}_\L \to \R, 
       \quad 
       (q_1,\ldots,q_n,p_1,\ldots,p_n) 
       \mapsto 
       \frac12\sum_{i =1}^n |p_i|^2 +  U(q_1,\ldots,q_n).
\end{split}
\end{equation}
Throughout this article we will assume the following for the interaction potential $U$:
\begin{itemize}
\item
$U$ is {\bf stable} in the sense that there is $L>0$ with 
\begin{equation}\label{stability}
     U(q_1,\ldots,q_n) \geq -nL, \quad n \in \N.
\end{equation}
\item
The potential energy of the empty or single particle configurations vanish: $U(\emptyset) = U(q_1) = 0$.
\item 
$U$ is shift invariant: $U(q_1,\ldots,q_n) = U(q_1+ h,\ldots,q_n+h)$, for all $h\in \R^3$, and all $n\geq 1$.
\item
$U$ is symmetric: $U(q_1,\ldots,q_n) = U(q_{\sigma(1)}, \ldots, q_{\sigma(n)})$ for any permutation $\sigma$ and all $n\geq 1$.
\end{itemize} 
{An abundance of examples of interaction potentials satisfying these assumptions can be found in \cite[pp.~34-39]{R}. These include pairwise interaction potentials of the form
\begin{equation}
\begin{split}
   U(q_1,\ldots,q_n)
   = 
   \sum_{i\neq j} \Phi(|q_i - q_j|),
\end{split}
\end{equation}
for appropriate $\Phi: [0, \infty) \to \R$. 
 Section \ref{Limits} here will require further assumptions on $U$.}

Define also the particle number and total momentum by
\begin{equation}
\begin{split}
       &N: \mathfrak{X}_\L \to \N, \quad (q_1,\ldots,q_n,p_1,\ldots,p_n) 
       \mapsto n, \\
      &\vec P: \mathfrak{X}_\L \to \R^3,  \quad (q_1,\ldots,q_n,p_1,\ldots,p_n) \mapsto \sum_{i=1}^n p_i.
\end{split}
\end{equation}
{In the above definitions, for the empty configuration, i.e.\,when  $n=0$, we take $(N,\vec P, H) = (0,\vec 0,0)$. }
For $(\aaa)$ in $\mathfrak I :=\R\times \R^3 \times \R_{<0}$,  define 
\begin{equation} \label{GCG}
\begin{split}
     g_{(\aaa),\L}(\tilde q,\tilde p)
    =
     \frac{1}{Z_{(\aaa),\L}}
     \exp
     \left(
     (\mu, \vec \l, \beta)
     \cdot
     (N, \vec P, H
     )
     (\tilde q,\tilde p)
     \right), 
\end{split}
\end{equation}
where
$Z_{(\aaa),\L}$ is the normalization constant (partition function):
\begin{equation}
\begin{split}
     Z_{(\aaa),\L} 
     = 
     \int_{\mathfrak{X}_\L} 
     \exp\left((\mu, \vec \l, \beta)\cdot(N, \vec P, H)(\tilde q,\tilde p)\right)
     \ \o(d\tilde q,d\tilde p).
\end{split}
\end{equation}
The above choice of $(\aaa)\in \mathfrak I$ is justified by the following simple lemma.
\begin{Lem} \label{Ihere}
Let $\mathfrak{I}_\L\subset \R \times \R^3 \times \R$ be the set of $(\aaa)$ such that $Z_{(\aaa),\L} < \infty$. Then
  $\mathfrak{I}_\L = \mathfrak{I}$.
  \end{Lem}
\proof
Complete the square for the $p$'s in the exponent and use the stability condition \eqref{stability} to see that $\beta<0$ is sufficient. 
Conversely, for $\beta \geq 0$, and regardless of $\vec \l$, the term of the integral for $N=1$ gives
\begin{equation} 
     \int_{\L\times \R^3} 
     \exp\left(\mu + \vec \l \cdot  \vec p + \beta \frac{|\vec p|^2}{2}\right)
     dq dp
     = + \infty.
\end{equation}
\qed

For each $(\aaa) \in \mathfrak I$, the probability measure $g_{(\aaa),\L}\o$ is the Gibbs grand-canonical distribution over $\L$ with parameters $(\aaa)$.
We will use $\mathbb{E}_{(\aaa),\L}[F]$ for the expectation of a function $F$ on $\mathfrak{X}_\L$ with respect to $g_{(\aaa),\L}\o$, i.e.,
\begin{equation}
\begin{split}
    \mathbb{E}_{(\aaa),\L}[F]
    =
    \int_{\mathfrak{X}_\L}
    F(\tilde q,\tilde p) \ 
    \frac{\exp\left(
    (\mu, \vec \l, \beta)\cdot(N, \vec P, H)(\tilde q,\tilde p) 
    \right)
}
    {Z_{(\aaa),\L}}
    \ \o(d \tilde q,d\tilde p).  
\end{split}
\end{equation}

\subsection{Information Entropy} 
For any $\nu\in \mathcal P_\L$, define the information (Gibbs) entropy by
\begin{equation}
        s_\L (\nu)=
        \begin{cases}
          \displaystyle \frac1{\V}\int_{\mathfrak{X}_\L} (\log f) f d\o 
           , & \nu\ll \o, f := \dfrac{d\nu}{d\o}  \\
           \\
         +\infty & {\text o/w}.
        \end{cases}
\end{equation}
In particular, if $f_n$ is the density of $\nu$ with respect to $\omega_n$,
i.e.\ $\nu = \sum_{n\geq 0} f_n \omega_n$,
then
\begin{equation}
\begin{split}
    s_\L (\nu) 
    =
    \frac1{\V}
    \sum_{n\geq 0}
    \ 
    \int\limits_{\L^n\times \R^{3n}}  
     \log f_n \ f_n d\omega_n.
\end{split}
\end{equation}
It is standard that $s_\L$ is bounded below by $0$ (using the convexity of the function $t\mapsto t\log t$ and Jensen's inequality), therefore it is a proper\footnote{Recall that a convex function is proper if it never takes the value $-\infty$ and it is not identically $+\infty$.} convex function on $\mathcal P_\L$. 

\subsection{Constrained Entropy Minimization on Fixed Volumes} \label{constrained mini on L}

In statistical mechanics it is standard to minimize the information entropy $s_\L$ over $\mathcal P_\L$ subject to suitable constraints 
\begin{equation} \label{constraints}
\begin{split}
        \int_{\mathfrak{X}_\L} \frac{N}{|\L|} d\nu = \rho,\quad
        \int_{\mathfrak{X}_\L} \frac{ \vec P}{|\L|} d\nu = \vec u, \quad
        \int_{\mathfrak{X}_\L} \frac{H}{|\L|} d\nu = E,
\end{split}
\end{equation}
and derive the grand canonical distributions as the minimizer.
For this, it is common to argue via Lagrange multipliers, see for example \cite[p.~66]{Z}. This is a formal application of the Lagrange multiplier theory on spaces of functions. 
{It is not at all clear to us how the Lagrange multipliers in infinite dimensions as for example in \cite[\S 43.8]{Zei}, \cite{BCM}, can be used to make such arguments rigorous.}{ (For example, one would need to choose a Banach space of functions in which the set of probability densities forms an open set.)}  
On the other hand, such arguments can find solid ground if one employs some differential structure in the space of measures, as for example in \cite{PS}. We shall come back to this point in future work. For the moment we recall with minor modifications the argument from \cite{K}: 
to minimize $s_\L$, it clearly suffices to look only at measures absolutely continuous with respect to $\o$, i.e.\,$\nu = f\o$. Assume there exists $(\mu, \vec \l, \beta)\in \mathfrak I$  such that 
\begin{equation} \label{mainEqn}
\begin{split}
     (\rho, \vec u, E)
      =
      \E
      \left[
      \frac{1}{|\L|} 
       (N, \vec P, H)
       \right], 
\end{split}
\end{equation}
an assumption that is not at all restrictive, see comments after equation \eqref{frak{S}L} below.
Then for $g_{(\aaa),\L}$ the grand canonical density as in \eqref{GCG},
by \eqref{mainEqn} 
we have that $\displaystyle \int_{\mathfrak{X}_\L}  f \log g_{(\aaa),\L} \ d\o$
is finite, therefore
\begin{equation}
\begin{split}
       |\L|  s_\L (f \o)
       &=
       \int_{\mathfrak{X}_\L} 
      f \log \left( \frac{f}{g_{(\aaa),\L}} \right) d\o
       +
       \int_{\mathfrak{X}_\L}
       f \log(g_{(\aaa),\L}) 
       \ d\o
       \\
       &=
        \int_{\mathfrak{X}_\L} 
      f \log \left( \frac{f}{g_{(\aaa),\L}} \right) d\o
      +
       (\mu, \vec \l,  \beta)\cdot(\rho, \vec u ,E)
       -
       \log Z_{(\aaa),\L} \\
       &\geq
       (\mu,  \vec \l,\beta)\cdot(\rho, \vec u , E)
       -       \log Z_{(\aaa),\L} 
       = |\L| s_\L(g_{(\aaa),\L} \o).
\end{split}
\end{equation}
In the last inequality we have used that the relative entropy $\displaystyle  \int_{\mathfrak{X}_\L}  f \log \left( \frac{f}{g_{(\aaa),\L}} \right) d\o$ is always greater than or equal to zero (by Jensen's inequality). As the relative entropy is zero if and only if $f = g_{(\aaa),\L}$ a.e., we conclude the standard

\medskip
\noindent
{\bf Entropy Minimization:}
For $(\rho,  \vec u, E)$ such that  there exists $(\aaa)$ 
with \eqref{mainEqn} satisfied the grand-canonical distribution $g_{(\aaa),\L} \o$ is the unique minimizer to $s_\L$
over all $\nu$'s in $\mathcal P_\L$ satisfying \eqref{constraints}.

Notice in the above process of entropy minimization the assumptions on the constrains $(\m)$: 
we require that $(\m) \in  \mathfrak{S}_\L $ where
\begin{equation} \label{frak{S}L}
\begin{split}
    \mathfrak{S}_\L 
    =
    \left\{
    (\rho, \vec u, E):
    {\rm there\ is}\ (\mu, \vec \l, \beta) 
    \in \mathfrak{I} 
    { \rm \ such \ that\
    \eqref{mainEqn}  \ holds}
    \right\}.
\end{split}
\end{equation}
In the next two subsections, we provide a detailed description of  $\mathfrak{S}_\L$ and will see that the assumption $(\m)\in \mathfrak{S}_\L$ is not restrictive at all: if $\displaystyle (\m) = \int_{\mathfrak{X}_\L} \frac{1}{|\L|} (N, \vec P, H) d\nu$ for any $\nu$ absolutely continuous probability measure with respect to $\o$ then $(\m)\in \mathfrak{S}_\L$, unless $\nu$ is the Dirac measure at the empty configuration, see Theorem \ref{prop: in the interior} and Theorem \ref{bijection}.


 \subsection{Homeomorphism between $\mathfrak I$ and $\mathfrak{S}_\L$}
%
%
%
We now introduce the log-partition function:
\begin{equation} \label{freenergyfinitevol}
        \Phi_\L(\aaa)=
         \begin{cases}
           \displaystyle 
           \frac1{|\L|}
           \log 
           Z_{(\aaa),\L}, &(\mu, \vec \l, \beta) \in \mathfrak{I}\\
           \\
        +\infty, & {\text o/w}.
        \end{cases}
\end{equation}
 A straightforward application of Fatou's lemma shows that $\Phi_\L$ is closed, in the sense that it has closed epigraph\footnote{Recall that the epigraph of $f$ is $\{(x,a): f(x) \leq a\}$.}.  A standard calculation shows that $\Phi_\L$ is strictly convex on $\mathfrak{I}$. Therefore $\nabla \Phi_\L$ defines a bijection
from $\mathfrak{I}$ to $\mathfrak{S}_\L$, see \cite[Theorem B, p.~99]{RV}. That this bijection is in fact a homeomorphism follows from 

\begin{Lem} \label{covex lemma}
Let $f$ be a proper and closed convex function with domain an open set. If $f$ is differentiable and strictly convex on its domain, then $\nabla f$ defines a homeomorphism on the image of $\nabla f$.
\end{Lem}
\begin{proof}
The domain of the subdifferential of $f$ is between intdom$(f)$ and dom$(f)$, see \cite[p.~253]{Rock}.
Since the domain is open, this means that the subdifferential is identical to $\nabla f$ on dom$(f)$ and empty everywhere else.
Then $f$ is essentially differentiable by \cite[Theorem 26.1]{Rock}.
As $f$ is also strictly convex, by assumption, the statement follows from \cite[Thorem 26.5]{Rock}.
\qed
\end{proof}
Notice now that 
for $(\aaa)$ in $\mathfrak{I}$
\begin{equation}
    \E
    \left[ 
    \frac{1}{|\L|}
    \left(N,  \vec P, H \right)\
     \right]
     =
     \frac{1}{|\L|}
     \nabla_{\aaa}
     \log Z_{(\aaa),\L}
     =
     \nabla_{\aaa} \ \Phi_\L.
\end{equation}
We therefore have:
\begin{Prop}  \label{global bijection}
$\mathfrak{S}_\L = {\rm Image}(\nabla \Phi_\Lambda)$ and the map 
\begin{equation}
\begin{split}
     (\mu, \vec \l, \beta)
     \mapsto
     \E
     \left[ \dfrac{1}{|\L|}
     \left({N}, \vec P, H\right)
     \right]
\end{split}
\end{equation}
defines a homeomorphism between $\mathfrak{I}$ and $\mathfrak{S}_\L$.
\end{Prop}

\subsection{Characterization of $\mathfrak{S}_\L$}
In this subsection, we describe the set $\mathfrak{S}_\L$.
 It is easy to find some rough bounds:
\begin{Lem}
$\mathfrak{S}_\L 
\subset 
\left\{
    (\rho, \vec u, E)\in \R \times \R^3 \times \R: 
    -L \rho \leq E, \rho >0 \right
    \}$, for $L$ the stability constant as in \eqref{stability}.
\end{Lem}
\proof
For $\displaystyle
     \rho = 
     \E
     \left[ 
     \dfrac{N}{|\L|}
     \right]$, using the stability condition and working as in the proof of Lemma \ref{Ihere}, it follows that $0< \rho < + \infty$. Furthermore
\begin {equation}
\begin{split}
     \E\left[\dfrac{1}{|\L|}
     H\right]
     \geq
     \E\left[\dfrac{1}{|\L|}
     U \right]
     \geq
     \E\left[\dfrac{1}{|\L|}
     (-L N)\right] 
     =
     -L \rho, 
\end{split}
\end{equation}
after using \eqref{stability} again. 
\qed

\subsubsection{Essential Range of $(N,\vec P, H)$}
Theorem \ref{bijection} below will describe $\mathfrak{S}_\L$ in more detail after some preparation.  
Let $f$ be a measurable map between some measure space $X$ to $\R^n$.
Recall first that the essential range of $f$  is all points $x\in \R^n$ such that $\omega(f^{-1}(B)) >0$, for any neighborhood $B$ of $x$. (Alternatively, the essential range is the support of the distribution of $f$.) When $f$ is a function on $\mathfrak X_\L$ or $\L^n\times \R^{3n}$, unless otherwise specified, we assume the reference measure to be $\o$ or $\o_n$, respectively.

For $\L$ fixed and for each $n\geq0$, use the stability of $U$ to define $L_\L^{(n)}$ as the smallest constant satisfying  
\begin{equation}
\begin{split}
    - n L_\L^{(n)} \leq U(q_1,\ldots,q_n), 
    \ 
     \text{for almost all }(q_1,\ldots,q_n) \in \L^n
     \text{ (w.r.t.\ Lebesgue \ measure \ on \ } \L^n).
\end{split}
\end{equation}
%

\begin{center}
\begin{figure}
\begin{center} 
\begin{tikzpicture}
[domain=0:1.,xscale=1.5,yscale=1.5]
\path [fill=gray!30!white] 
    (0,0) 
   -- 
    (-1.1,1.11) 
    -- 
    (-1.2,1.91)--(-1,1.9)
    to 
    [out = -77, in= 135] (0,0);
\path [fill=gray!30!white] 
(-1,1.9) to [out = -77, in= 135] (0,0)--(2,0)--(2,1.9)--(-1,1.9);
\draw [gray!30!white] (-1,1.9)
 to [out = -77, in= 135] (0,0);
\draw[-] (0,0) 
    -- 
    (-1.1,1.11);
\draw[dashed] (0,0)-- (0,.63) node[right] {$- L^{(1)}_\L/ |\L|$};
\draw[fill] (0,.63) circle [radius=0.025];
\draw[dashed] (0,.63)--(-1.1,1.1);
\draw[fill] (-1.1,1.1) circle [radius=0.025]node[right] {$-2 L^{(2)}_\L/ |\L|$};
\draw[-]  (0,0)--(-1.2,0) ; 
\draw[domain=0:-2.]
      plot (\x, {-0.5*\x}) node[left]{$-L \rho$} ;
\draw[domain=-1.1:-1.2]
      plot (\x, {-8.*(\x + 1.1)- 1.2*(-1.1+ .7)+ .7*.9}) 
      ;
 \draw[<->] (0,2.5) node[left]{$\rho$}-- (0,0) -- (2.5,0) node[below] {$E$};
\end{tikzpicture}
\caption{The slice of $\mathfrak{S}_\L$ at $\vec u = \vec 0$. \label{Fig1}}
\end{center}
\end{figure}
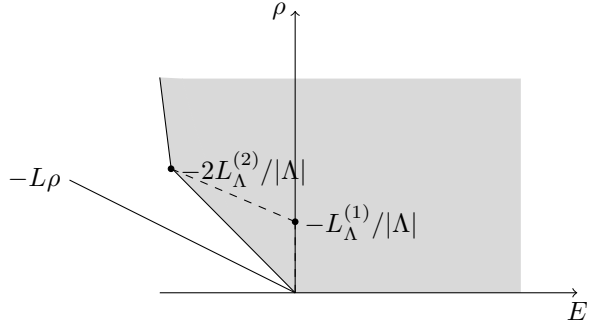 
\end{center} 

\begin{Notation}\label{notation}
From now on $\mathcal{E}_\L$ denotes the essential range of $(N,\vec P, H)/|\L|$ over $\mathfrak{X}_\L$, $\mathcal{C}_\L$ the convex hull of $\mathcal{E}_\L$, and $\mathring{\mathcal{C}}_\L$ the interior of $\mathcal{C}_\L$.

Similarly, for each $n\geq 0$, $\mathcal{E}_{\L,n}$ denotes the essential range of $(N,\vec P, H)/|\L|$ over $\L^n\times \R^{3n}$ and $\mathcal{C}_{\L,n}$ the convex hull of $\mathcal{E}_{\L,n}$. Notice that  $\mathcal C_\L$ is the convex hull of $\bigcup_{n\geq 0} \mathcal C_{\L,n}$. {Clearly, when $n=0$, $\mathcal{C}_{\L,0}=\mathcal{E}_{\L,0} = \left\{ (0,\vec 0, 0) \right\}$. }
\end{Notation}

\begin{Lem} \label{boundary is essential range}
For each $n\geq 1$, 
$\mathcal C_{\L,n} = \left\{
(\m)/|\L|: \rho=n, E\geq -nL_{\Lambda}^{(n)},
|\vec u|^2 \leq 2n(E+ nL_{\Lambda}^{(n)})
\right\}$.
\end{Lem}
\begin{proof}
By definition of $L_\L^{(n)}$, for $\o_n$-almost all $(\tilde q,\tilde p)\in \L^n\times \R^{3n}$ we have
\begin{equation}
 H(\tilde q,\tilde p)
 =
 \frac12\sum_{i =1}^n |p_i|^2 +  U(q_1,\ldots,q_n)
 \geq 
 \frac12\sum_{i =1}^n |p_i|^2 -nL_\L^{(n)}.
\end{equation}
Using 
$\left| \sum_{k=1}^n p_i \right |^2 \leq n \sum_{i =1}^n |p_i|^2 $ we have that  
\begin{equation}
\mathcal{E}_{\L,n}
\subset 
\left\{(\m)/|\L|: 
\rho = n, 
E\geq -nL_{\Lambda}^{(n)},
|\vec u|^2 \leq 2n(E+ nL_{\Lambda}^{(n)})
\right\}.
\end{equation}
Let 
\begin{equation}
\mathcal I_{\Lambda,n}
=
\left\{
(\m):
\rho=n,
E\geq -nL_{\Lambda}^{(n)},
|\vec u|^2 = 2n(E+ nL_{\Lambda}^{(n)})
\right\}.
\end{equation}
It remains to prove that $ \mathcal I_{\Lambda,n}$ belongs to the essential range of $(N,\vec P, H)$.

Take any $(n,\vec u_0,E_0) \in  \mathcal I_{\Lambda,N}$. 
Fix any $\epsilon>0$.
We will show that the inverse image of the set
\begin{equation}
S_\epsilon = \left\{
(n,\vec u,E): |\vec u - \vec u_0| \leq \epsilon,  |E - E_0|\leq \epsilon
\right\}
\end{equation}
has nonzero $\omega_n$-measure. 

As $-nL_{\Lambda}^{(n)}$ is the essential infimum of $U(q_1,\ldots,q_n)$, we know that 
\begin{equation}
A_\epsilon = 
\left\{
(q_1,\ldots,q_n)\in \Lambda^n:
0\leq U(q_1,\ldots,q_n) + n L_{\Lambda}^{(n)} 
\leq  \dfrac {\epsilon}{2}
\right\}
\end{equation}
has nonzero Lebesgue measure.
Let 
\begin{equation}
B_\epsilon = 
\left\{
(p_1,\ldots,p_n)\in \R^n:
\left | \sum_k p_k - \vec u_0\right | \leq \epsilon,
\text{ and }
\left| \dfrac12\sum_{k} |p_k|^2 - \dfrac{|\vec u_0|^2}{2n} \right| \leq \dfrac {\epsilon}{2}
\right\}.
\end{equation}
Notice that the image of $A_\epsilon \times B_\epsilon$ under
$(N, \vec P, H)$ is a subset of $S_\epsilon$. It remains to show that $B_\epsilon$ has nonzero Lebesgue measure.
In fact, take $\tilde p_0 \in \R^{3n}$ such that 
\begin{equation}
\tilde p_0 
=
\dfrac 1n (\vec u_0,\ldots,\vec u_0).
\end{equation}
Clearly, $\tilde p_0 \in B_\epsilon$.
As $\sum_{k=1}^n p_k$ and $\sum_{k=1}^n |p_k|^2$ are continuous functions on $\R^{3n}$, then there exists $\delta>0$ such that 
\begin{equation}
B_{\epsilon,\delta}
:=\left \{ \tilde p \in \R^{3n} : 
|\tilde p - \tilde p_0 | \leq \delta
\right\}
 \subset B_\epsilon.
\end{equation}
As $B_{\epsilon,\delta}$ has nonzero Lebesgue measure, the proof is now complete.
\qed
\end{proof}

\begin{Rmk}
In fact, for $n\geq 2$, the essential range $\mathcal{E}_{\L,n}$ equals the convex hull $\mathcal{C}_{\L,n}$. The proof is similar to that of Lemma \ref{boundary is essential range}. As this is not relevant to the rest of the article, we omit the details.
\end{Rmk}

Let $\overline {\mathcal C_\L}$ be the closure of $\mathcal C_\L$. As $\displaystyle \mathcal{C}_\L = {\rm conv}\left(\bigcup_{n\geq 0} \mathcal{C}_{\L,n} \right)$, from the structure of $\mathcal C_{\L,n}$ we see that if $(\rho, \vec u, E)\in \overline {\mathcal C_\L}$, then $(\rho, \vec u, E+a)\in \overline {\mathcal C_\L}$ for any $a>0$. 
A useful consequence of this is: 
for any $(\rho_0, \vec u_0, E_0)\notin \overline {\mathcal C_\L}$,
we can find $(\alpha_0, \vec \alpha, \alpha_4)$
with $\alpha_4 > 0$ 
and $\varepsilon >0$ so that for all $(\m) \in \overline {\mathcal C_\L}$
\begin{equation} \label {eqn: out of convex hull separation}
\begin{split}
      \left[(\m) - (\rho_0, \vec u_0, E_0)\right]
      \cdot
      (\alpha_0, \vec \alpha, \alpha_4)
       > \varepsilon.
\end{split}
\end{equation}

\subsubsection{Results on $\mathfrak S_\L$}
Now we are ready to state
\begin{Thm} \label{bijection}
The set of solvability $\mathfrak{S}_\L$ is the interior of the convex hull of the essential range $\mathcal{E}_\L$, i.e.
$\mathfrak{S}_\L = \mathring{\mathcal{C}}_\L$.
\end{Thm} 
The rest of this subsection will give a proof for this theorem. To this end, we first prove a result for the essential range of $(N,\vec P,H)$ under more general probability measures on $\mathfrak{X}_\L$.

 \begin{Lem} \label{Lem: interior for N}
For $n \geq 1$ and $\nu$ any probability measure on $\L^n \times \R^{3n}$ which is absolutely continuous with respect to $\omega_n$,
we have that the average of $(N,\vec P, H)/|\L|$ under $\nu$ is in the {relative} interior of $\mathcal C_{\Lambda,n}$ 
(i.e. $\mathcal C_{\Lambda,n}$ is treated as a subset in $\R^4$) and therefore in the interior of $\mathcal C_{\Lambda}$ (as a subset in $\R^5$).
\end{Lem}
\begin{proof}
Let $\displaystyle Q_0 = (n,\vec P_0, H_0) = \int_{\L^n\times \R^{3n}} (N,\vec P, H) d\nu$. 
Assume that $Q_0/|\L|$ is on the boundary
 of $\mathcal C_{\Lambda,n}$ (as a subset of $\R^4$), then
$Q_0 \in \mathcal I_{\Lambda,n}$. By the strict convexity of $\mathcal C_{\Lambda,n}$, 
there exists $\vec \alpha \in \R^5 $ such that for all $Q \in\mathcal C_{\Lambda,N}$ such that $Q\neq Q_0$
\begin{equation}  \label {C_n convexity}
\vec \alpha \cdot (Q - Q_0) >0.
\end{equation}
Notice that the inverse image of $Q_0$ with respect to the map $(N,\vec P, H)$ is a subset of the set
\begin{equation}
\left\{
(\tilde q, \tilde p) \in \Lambda^n \times \R^{3n}:
\sum_{k=1}^n p_k = \vec P_0
\right\}
\end{equation}
which has zero $\omega_n$-measure, therefore also zero $\nu$-measure. We obtain now
\begin{equation} \label {integral > zero, convexity}
\int_{\L^n\times \R^{3n}}  \vec \alpha \cdot ( (N,\vec P, H) - Q_0) d\nu >0,
\end{equation}
clearly contradicing the fact $\displaystyle Q_0 = (n,\vec P_0, H_0) = \int_{\L^n\times \R^{3n}}  (N,\vec P, H) d\nu$. 

On the other hand, if $Q_0$ is in the exterior of $\mathcal C_{\Lambda,n}$ (as a subset of $\R^4$), by convexity of $\mathcal C_{\Lambda,n}$, we can find $\vec \alpha$ such that \eqref{C_n convexity} holds for all $Q \in\mathcal C_{\Lambda,N}$. Therefore we have \eqref{integral > zero, convexity} which leads again to a contradiction.  
\qed
\end{proof}

\begin{Thm} \label {prop: in the interior}
Let $\nu$ be any probability measure on $\mathfrak X_\L$ absolutely continuous with respect to $\omega$. 
Assume that $\nu(\{\emptyset \}) < 1$, i.e. $\nu$ is not the Dirac measure on the empty configuration.  
Then the average of $(N,\vec P, H)/|\L|$ is in the interior of $\mathcal C_{\Lambda}$.
\end{Thm}
\proof
Let $Q =( \rho,\vec u, E)$ be the average of $(N,\vec P,H)/|\L|$ with respect to $\nu$. 
As $\nu$ is not the Dirac at the empty configuration, 
we know that $Q$ is a convex combination of some family of $\{Q_{n_k}\}_k$, $n_k\geq 1$, for $Q_n$ is the average of $(N,\vec P,H)/|\L|$ conditioned on configurations of $n$ particles.

From Lemma \ref{Lem: interior for N}, we know that all $Q_{n_k}$ are in the interior of $\mathcal C_{\Lambda}$, and then so is their convex combination $Q$.
\qed

For any $(\aaa)\in \mathfrak I$, apply  Theorem \ref{prop: in the interior} to the grand-canonical distribution $g_{(\aaa),\L}\o$ to conclude that $\mathfrak{S}_\L \subset \mathring{\mathcal{C}}_\L$. To prove Theorem \ref{bijection}, it remains to show that for any $(\m)\in \mathring{\mathcal{C}}_\L$, there exists $(\aaa)\in \mathfrak{I}$ such that \eqref{mainEqn} holds.
The contents of the following Lemma appear in \cite[Lemma A4.6]{L} without proof:
\begin{Lem} \label{LanfordLemma}
For any $(N_0, \vec P_0, H_0)\in |\L| \mathring{\mathcal{C}}_\L$, i.e.\,the interior of the convex hull of the essential range of $ (N, \vec P, H)$ over $\mathfrak X_\L$, 
there is $\varepsilon$ such that for all unit vectors $\vec e$ in $\R^5$ the following holds:
\begin{equation}\label{Lanford's}
\begin{split}
      \o\left(
       \left\{
       (\tilde q, \tilde p) \in \mathfrak{X}_\L:
       \left[(N, \vec P, H)(\tilde q, \tilde p) - (N_0, \vec P_0, H_0) \right]
       \cdot 
       \vec e
       \ > \varepsilon\
       \right\}
       \right)
       > \varepsilon.
\end{split}
\end{equation}
\end{Lem}
\proof
 Whenever $(N_0, \vec P_0, H_0)\in |\L| \mathring{\mathcal{C}}_\L$, for any  unit vector $\vec e$ there is some $(N_*, \vec P_*, H_*)$ in the essential range of $(N, \vec P, H)$ that belongs to the half-space through $(N_0, \vec P_0, H_0)$ with (inward) normal $\vec e$; otherwise the whole convex hull of the essential range would be on one side of $(N_0, \vec P_0, H_0)$, i.e.~$(N_0, \vec P_0, H_0)$ would not be in its interior. 

We then have 
\begin{equation}
     v 
     := 
     \left[
     (N_*, \vec P_*, H_*)- (N_0, \vec P_0, H_0)
     \right] \cdot \vec e > 0.
\end{equation}     
Let 
$\mathcal{N}$ be a neighborhood of $(N_*, \vec P_*, H_*)$ such that for all $(N', \vec P', H')$ in $\mathcal{N}$ we have
\begin{equation}
\begin{split}
     \left[
     (N', \vec P', H')- (N_0, \vec P_0, H_0)
     \right] \cdot \vec e
     > \frac{v}{2}.
\end{split}
\end{equation}
Then the $\o$ measure of the inverse image of $\mathcal{N}$ is not zero, say $m>0$.
Take $\varepsilon = \min\{v/2, m\}$ and notice that this $\varepsilon$ satisfies \eqref{Lanford's} for $\vec e$. 

By continuity of projection, for the {same} $\mathcal{N}$, there is a neighborhood $\mathcal{U}$ of $\vec e$ on the unit sphere such that for  
any unit vector $\vec e\, '$ in $\mathcal{U}$ and any $(N', \vec P', H')$ in $\mathcal{N}$ we have
\begin{equation}
\begin{split}
     \left[
     (N', \vec P', H')- (N_0, \vec P_0, H_0)
     \right] \cdot \vec e'
     > \frac{v}{4}.
\end{split}
\end{equation}
Redefining $\varepsilon = \min\{ v/4, m \}$, we now have \eqref{Lanford's} satisfied on $\mathcal{U}$.

By compactness of the unit sphere, we only need to repeat this finitely many times and take the smallest $\varepsilon$ to make sure that \eqref{Lanford's} is satisfied for all unit vectors.
\qed

We are ready to show that for any $(\m)\in \mathring{\mathcal{C}}_\L$, there exists $(\aaa)\in \mathfrak{I}$ such that \eqref{mainEqn} holds.
This step builds on \cite[pp.~44-47]{L}, see also \cite[\S 16]{Kh}.  For $(\rho, \vec u, E) \in  \mathring{\mathcal{C}}_\L$ fixed, let $K: \mathfrak{I} \to \R_{>0}$ be given by
\begin{equation}
\begin{split}
     K: (\mu, \vec\l, \beta)
     \mapsto
    \int_{\mathfrak{X}_\L} 
     \exp
     \left((\mu,\vec\l,\beta) 
     \cdot
     (N - \rho|\L|,\ \vec P - |\L|\vec u,\ H - |\L|E\ )\right)\
     d\o.
\end{split}
\end{equation}
{Clearly, $\log K = |\L| \left [  \Phi_\L - (\aaa)\cdot (\m) \right]$ and therefore
$\log K$ is strictly convex on $\mathfrak{I}$.
Note that
\begin{equation}
\begin{split}
      \nabla_{\aaa} \log K = 0
      \Leftrightarrow
      (\m) = \nabla_{\aaa} \Phi_\L (\aaa).
\end{split}
\end{equation}
To conclude there is unique critical point of $\log K$ in $\mathfrak{I}$, it suffices to show that $K$, and therefore $\log K$, goes to $+\infty$
as $(\aaa)$ approaches the boundary of $\mathfrak{I}$. Precisely, we show that $K(\aaa) \to \infty$ if $|(\aaa)| \to \infty$ or $(\aaa)\to (\mu_0,\vec \l_0,0)$ for any finite $\mu_0$ and $\vec \l_0$. 

To this end, we first apply Lemma \ref{LanfordLemma} for
 \begin{equation}
 (N_0, \vec P_0, H_0) =  |\L|(\rho, \vec u, E)
  \text{ and } 
 \vec e = (\mu,\vec\l,\beta)/|(\mu,\vec\l,\beta)|
 \end{equation}
 to find $\varepsilon>0$ such that for any $(\aaa)$, there exists a set of $\mathfrak{X}_\L$ of $\o$ measure at least $\varepsilon$ where 
\begin{equation}
\begin{split}
     (\mu,\vec\l,\beta)\cdot
     \left( 
     (N, \vec P, H) - |\L|(\rho, \vec u, E)
     \right) 
     > 
     \varepsilon |(\mu,\vec\l,\beta)| .
\end{split}
\end{equation}
It follows that
\begin{equation}
\begin{split}
     K(\mu, \vec \l, \beta) \geq 
     \varepsilon\ 
     \exp\left( |(\mu, \vec \l, \beta)|\varepsilon\right). 
\end{split}
\end{equation}
Therefore $K(\mu, \vec \l, \beta) \to + \infty$, as $|(\mu,\vec \l, \beta)| \to \infty$ while staying in $\mathfrak{I}$, and so does $\log K(\mu, \vec \l, \beta)$. 

Next, if $(\mu_m, \vec \l_m,\beta_m) \to (\mu, \vec \l, 0)$, as $m \to \infty$, Fatou's lemma gives
\begin{equation}
\begin{split}
     &\liminf_{m \to \infty}\int_{\mathfrak{X}_\L}
     \exp\left((\mu_m,\vec \l_m,\beta_m) 
     \cdot
     (N - \rho|\L|,\ \vec P - |\L|\vec u,\ H - |\L|E)\right)
     d\o\\
     \geq
     & \int_{\mathfrak{X}_\L}
     \exp\left((\mu,\vec \l) 
     \cdot
     (N - \rho|\L|,\ \vec P - |\L|\vec u)\right)\
     d\o = + \infty,
\end{split}
\end{equation}
showing that $\log K$ explodes when $(\aaa)$ approaches the hyperplane $\beta = 0$. 
The proof of Theorem \ref{bijection} is now complete.}

\subsection{Information Entropy and log-Partition Function} \label{conjugateoffree} We shall compare eventually the information entropy $s_\L$ with the entropy that appears at the thermodynamic limit using convex conjugates. For this define now
the convex conjugate of $\Phi_\L$ from \eqref{freenergyfinitevol}, 
\begin{equation}
\begin{split}
     \Phi^*_\L(\m) 
     = 
     \sup_{(\aaa)\in \R^5}
     \left[ 
     (\m) \cdot (\aaa) - \Phi_\L(\aaa)
     \right].     
\end{split}
\end{equation}
(For standard convex analysis consult \cite{Rock}.)
The following describes the relation of $\Phi_\L$ to $s_\L$.
\begin{Prop} \label{PhiStar}
$\Phi^*_\L$ takes values as follows:
\begin{enumerate}
\item \label{first}
For $(\rho, \vec u, E)$ in $\mathfrak{S}_\L$,
\begin{equation}
\begin{split} \label{entropyAsConjugate}
   \Phi^*_\L(\m) 
   = 
   s_\L
   \left( 
   g_{(\aaa),\L} 
   \o
   \right)
   =
   (\m) \cdot (\aaa) - \Phi_\L(\aaa),
\end{split}
\end{equation}
where $(\aaa)\in \mathfrak I$ is the unique solution of 
\begin{equation}  \label{basic equation}
\begin{split}
       \E \left[
       \frac{(\M)}{|\L|} \right] 
       = 
       (\rho, \vec u, E).
\end{split}
\end{equation}
\item  \label{second}
For $(\rho, \vec u, E) \notin \overline{\mathfrak{S}_\L}$,  $\Phi^*_\L(\m) =  +\infty$ .
\item \label{third}
And for $(\rho, \vec u, E)  \in \partial (\mathfrak{S}_\L)$, $\Phi^*_\L(\m)$ is the limit along any interior segment as we approach $(\m)$:
\begin{equation}
\begin{split}
      \Phi^*_\L(\m)
      =
      \lim\limits_{t\to 1} 
     \Phi^*_\L((1-t) (\rho', \vec u', E')+ t(\rho, \vec u, E)),
\end{split}
\end{equation}
for any  $(\rho', \vec u',E')\in \mathfrak{S}_\L$. 
\end{enumerate}
\end{Prop}
\proof
For \eqref{first} notice that the critical point equation of  
\begin{equation}
\begin{split}
    (\mu, \vec \l, \beta) 
    \mapsto
    (\m) \cdot (\aaa) - \Phi_\L(\aaa) 
\end{split}
\end{equation}
is
\begin{equation}
\begin{split}
    (\m) = \nabla_{(\aaa)} \Phi_\L(\aaa)
    =
    \E\left[ \frac1{|\L|}(N, \vec P, H)\right], 
\end{split}
\end{equation}
which, by Theorem \ref{bijection} has solution if and only if $(\rho, \vec u, E)$ is in $\mathfrak{S}_\L$.

For \eqref{second}, when $(\rho, \vec u, E)$ is not in the $\overline{\mathfrak{S}_\L}$, by \eqref{eqn: out of convex hull separation}, we find $(\alpha_0, \vec \alpha, \alpha_4)$
with $\alpha_4 > 0$ 
and $\varepsilon >0$ so that almost always
\begin{equation}
\begin{split}
      \left[(N, \vec P, H) - |\L|(\rho, \vec u, E)\right]
      \cdot
      (\alpha_0, \vec \alpha, \alpha_4)
       > \varepsilon.
\end{split}
\end{equation}
Notice that for any $(\aaa)\in \mathfrak{I}$,
\begin{equation}
\begin{split}
&(\m) \cdot (\aaa) - \Phi_\L(\aaa) \\
&\quad \quad=
-\frac1{\V}  \log  
 \int_{\mathfrak{X}_\L}   
        \exp  \big\{ 
         \left[(N, \vec P, H) - |\L|(\rho, \vec u, E)\right]
      \cdot
      (\aaa)
         \big\} d\o .
\end{split}
\end{equation}
Then for $\vartheta>0$,
 with choosing 
 $(\aaa) = - \vartheta (\alpha_0, \vec \alpha, \alpha_4) - (0,\vec 0, 1)$,
 we have
\begin{equation}
\begin{split}
 \Phi^*_\L(\m) 
 &\geq (\m) \cdot (\aaa) - \Phi_\L(\aaa) \\
        &\geq
        -\frac1{\V} 
        \log 
        \int_{\mathfrak{X}_\L}   
        \exp\left\{-\vartheta \varepsilon -(H- |\L|E)\right\} \ d\o\\
        &=
        \frac1{\V}\vartheta \varepsilon
        -
        \frac1{\V} 
        \log 
        \int_{\mathfrak{X}_\L}   
        \exp\left\{ -(H- |\L|E)\right \}d\o.
\end{split}
\end{equation}
Letting $\vartheta \to \infty$, we conclude that $\Phi_\L^* = +\infty$ whenever $(\m)$ is not in $\overline{\mathfrak{S}_\L}$.
 
For \eqref{third}, when $(\m)$ is on the boundary of the convex hull, $\Phi_\L^*(\m)$ is the limit of the values along any segment that has $(\m)$ as end point and lies in the interior of the convex hull otherwise, as is the case for any convex closed function, see \cite[Theorem 7.5]{Rock}. (For similar results in a different setting and with different proofs, see \cite[Theorem 3.4.]{WJ}.)
\qed

\section{Thermodynamic Limits}\label{Limits}
We now recall the definition of microcanonical thermodynamic limit entropy.
In the next section we will see how it relates to information entropy on bounded domains. This section uses Martin-L\"of's formalism in \cite{M-L} for thermodynamic limits as the volume of the domain becomes infinite, see also \cite{L}\footnote{Note that the formulas in \cite{L} and \cite{M-L} are related as follows: the microcanonical thermodynamic limit entropy at $(\rho, E)$ in \cite{M-L} equals the microcanonical thermodynamic limit entropy in \cite{L} at $(\rho, E/\rho)$ multiplied by $\rho$.}. 

As already mentioned, we shall compare entropies via their convex conjugates. We examined convex conjugates of information entropy in the previous section. For the microcanonical thermodynamic limit entropy it is well known that its convex conjugate is related to the (grand canonical) thermodynamic limit pressure, the limit of the log-partition functions $\Phi_\L$'s as $\L$ increases. We are especially interested in the strict convexity of pressure: we use it in this section for the differentiability of the limit entropy (which we shall use in the main result of the next section) and we also use it to show a one-to-one correspondence between thermodynamic parameters and macroscopic quantities at the thermodynamic limit in section \ref{localbijection}.

\subsection{Microcanonical Entropy}
We first introduce the temperedness condition for interaction potential:
$U$ is called {\bf tempered} if for some constants $\delta$ greater than the space dimension, $K>0$, and $R>0$ it holds that, for $\tilde q_1$ and $\tilde q_2$ configurations consisting of $N_1$ and $N_2$ position points respectively,
\begin{equation} \label {condition: temperedness}
\begin{split}
    \left|U\left(\tilde q_1, \tilde q_2\right)
    -
    U\left(\tilde q_1\right) - U\left(\tilde q_2\right)\right|
    \leq
    K \frac{N_1 N_2}{ \left( d\left(\tilde q_1, \tilde q_2\right) \right)^\delta}, 
\end{split}
\end{equation}
whenever $\tilde q_1$ and $\tilde q_2$  have distance $d\left(\tilde q_1, \tilde q_2\right) >R$. (Lennard-Jones type potentials and finite range interactions are tempered.)

To define the microcanonical entropy, we will need to specify a special sequence of $\L$'s that increases to infinity.

\begin{Notation} \label{notation2}
For any $l \in \N$, $\L_l$ will denote a box in $\R^3$ with sides of length $2^l$. And $\L'_l \subset \L_l$ will denote a smaller box with side $2^l - 2 R_l$, where $R_l = R_0 2^{\rho l}$, for some fixed $R_0$ and $\rho \in [0,1)$. (The choice of $R_0$ and $\rho$ depends on the interaction potential $U$ and $\rho = 0$ is when $U$ is of finite range, cf. \cite[p.~105, p.~88]{M-L}.)  
\end{Notation}

For any $(\m)\in \R^5$, we fix a sequence of open convex sets $\left\{A_k\right\}_{k\in \N}$ such that $A_k$ shrinks to $(\m)$. For interactions potentials which are both stable and tempered, the microcanonical entropy at $(\m)$ is defined as
\begin{equation}
s(\m) =\lim_{k\to \infty} \lim_{l\to \infty}  
\frac{1}{|\L_l|} \log 
 \o\left(
       \left\{
       (\tilde q, \tilde p) \in \mathfrak{X}_{\L'_l}:
      \dfrac {1}{|\L_l |} (N, \vec P, H)(\tilde q, \tilde p) \in A_k
       \right\}
       \right),
\end{equation}
cf.~\cite[\S 3.4.2]{M-L}.
We also know that $s(\m)$ is an upper-semi continuous concave function on $\R^5$ and that 
$s(\m)<\infty$, cf.~\cite[p.~45,p.~96]{M-L}. 
Therefore the set ${\rm dom}(s) = \{(\m): s(\m) > - \infty \}$, i.e.\,the domain of $s$, is convex. We will give a description of {\rm int\,dom}(s), i.e.\,the interior of domain of $s$, in Theorem \ref{interiors}. To this end, we will need the following notation:
\begin{Notation} \label{notation 3}
$\mathcal{E}'_{\L_l}$ will denote the essential range of $\displaystyle \frac1{|\L_l|}(N,\vec P, H)$ as a map on $\frak{X}_{\L'_l}$. (Compared with $\mathcal{E}_{\L_l}$ as in Notation \ref{notation}, $\mathcal{E}'_{\L_l}$ takes into account only configurations in $\L'_l$.) Also $\rm{conv} (\mathcal{E}'_{\L_l})$ denotes the convex hull of $\mathcal{E}'_{\L_l}$ and $\rm{int\, conv} (\mathcal{E}'_{\L_l})$ is the interior of $\rm{conv} (\mathcal{E}'_{\L_l})$.  As $\mathcal{E}'_{\L_l} = \dfrac{\L'_l}{\L_l} \mathcal{E}_{\L'_l}$, from the proof of Lemma \ref{boundary is essential range}, it is easy to see that 
\begin{equation} \label {conv prime vs without}
\rm{conv} (\mathcal{E}'_{\L_l})
=
\dfrac{|\L'_l|}{|\L_l|} \rm{conv} (\mathcal{E}_{\L'_l})
=
\dfrac{|\L'_l|}{|\L_l|} \mathcal C_{\L'_l}.
\end{equation}
\end{Notation}
In this notation, the following will be crucial in the proof of Theorem \ref{interiors}: 
\begin{Lem}
For stable, tempered interaction potentials the 
domain of $s$ is related to $\mathcal{E}'_{\L_l}$ as follows:
\begin{equation} \label{closure union}
\begin{split}
    \overline{{\rm dom}(s)}
    =
    \overline{\,
    \bigcup_l \mathcal{E}'_{\L_l}}.
\end{split}
\end{equation}
\end{Lem}
\begin{proof}
See \cite[p.~93, p.~95]{M-L} for finite range, stable interaction potentials and \cite[pp.~105--111]{M-L} for infinite range, stable, tempered interaction potentials. Note that the claim there is for $\mathcal{E}_\L$ instead of $\mathcal{E}'_{\L_l}$ (in our notation).\qed
\end{proof}


\subsection{Pressure} \label{Pressure}
For any $(\aaa) \in \R^5$, define now the thermodynamic limit pressure by 
\begin{equation}  \label{PhitoP}
\begin{split}
     \Xi(\aaa) 
      =
      \lim_{\L \to \infty} \Phi_\L(\aaa)
\end{split}
\end{equation}
where the sequence of $\L$'s approaches infinity in the sense of 
 ``strong van Hove" \cite[p.~107]{M-L} and ``approximable by cubes" \cite[p.~91]{M-L}. The sequence $\{\L_l\}$ as in Notation \ref{notation2} is an example of such $\L$'s.
See \cite[p.~111]{M-L} for the existence and properties of this 
limit\footnote{In \cite{M-L}'s notation, our $\Xi(\mu, \vec \l, \beta)$ is $s(\R \times \R^3 \times \R_{>0}; \mu,\vec \l, \beta)$. 
\cite{M-L} also comments on differentiability of $\Xi$ for finite range interactions but leaves open the differentiability for stable potentials. We therefore resort to our own differentiability arguments here.}. 
Trivially, by \eqref{PhitoP}, $\Xi(\aaa) = \infty$ for $(\aaa)\notin \mathfrak I$.
{
\begin{Lem}\label{domP} 
For $(\aaa)\in \mathfrak{I}$, we have $0\leq \Xi(\aaa) < \infty$.
(Therefore $\Xi$ is a proper convex function with domain ${\rm dom} (\Xi) = \mathfrak I$.) 
\end{Lem}
\proof
Fix any $(\aaa)$ in $\mathfrak{I}$.
By the stability \eqref{stability} and completing squares for the velocity terms, we easily see that $\Xi(\aaa) < \infty$.
On the other hand, as $Z_{(\aaa),\L}\geq1$  for all $\L$, we have  $\Phi_\L(\aaa) \geq 0$, therefore $\Xi(\aaa) \geq 0$.
\qed
}

The thermodynamic limit microcanonical entropy $s$ and the pressure $\Xi$ are related via convex conjugation:  
\begin{equation} \label{pressure} 
\begin{split}
      \Xi(\aaa) 
     =
      \sup_{(\m)\in \R^5} \{   s(\m) + (\m)\cdot(\mu, \vec \l, \beta)\},
\end{split}
\end{equation}
(see \cite[p.~45, Lemma 5]{M-L}\,\footnote{To go from Martin-L\"of's Lemma 5 to \eqref{pressure} notice that the limit on the left of \eqref{pressure} is the supremum of what Martin-L\"of calls $s(u, a)$ in Lemma 5, p.~45, and that the {\em pointwise} formula (in Martin-L\"of's notation) $s(u,a) = s(u) - a\cdot u$ holds by \cite[Lemma 5a]{M-L}.})
i.e.~
\begin{equation} \label{Xi-s}
\begin{split}
      \Xi(\aaa)
      =
      (-s)^*(\mu, \vec \l, \beta), 
\end{split}
\end{equation}
with $^*$ still denoting the convex conjugate. {As $s$ is concave and upper-semi continuous, we can conjugate once again to get
\begin{equation} \label {-s as a conjugate}
\begin{split}
      \Xi^*(\rho, \vec u, E) = -s(\rho, \vec u, E).
\end{split}
\end{equation}
}

As mentioned in the introduction to this section, it will be important to know when $\Xi$ is strictly convex. 
We assume strict convexity for the moment and provide a discussion in section \ref{section: strict convexity}.

\begin{Lem} \label{s essentially differentiable}
Assume that $\Xi(\aaa)$ is strictly convex on $\mathfrak{I}$. Then $s$ is essentially differentiable, i.e.~differentiable in the interior of its domain and $|\nabla s| \to \infty$ on the boundary of its domain. 
\end{Lem}
\proof
As $\text{\rm dom} (\Xi) = \mathfrak{I}$, we have that $\Xi$, as an extended-reals-valued function, is essentially strictly convex (for the general definition see \cite[p.~253]{Rock}).
This observation, \eqref {-s as a conjugate}, and \cite[Theorem 26.3]{Rock} immediately imply the lemma (which should be compared to \cite[Remark 3.7, Remark 6.5]{Gii} where differentiability of $s$ on its domain was shown).
\qed

\subsection{Strict Convexity of Pressure} \label {section: strict convexity}
Strict convexity of pressure for lattice systems was shown by Griffiths and Ruelle in \cite{GR}.
For continuous systems, this is the case in general when the Gibbs variational principle (see \cite{D}, \cite{Gii} for a definition) holds.
As here is the only place we use the concept of Gibbs states, we only sketch the argument. For details see \cite[Proposition 8.5]{Preston} and \cite[Remark 3 and \S 5]{Gi}.

Let $\mathfrak X$ be the the space of locally finite configurations on $\R^3$ marked with velocities. For a state $\nu$ (a probability measure on $\mathfrak X$), one can define specific entropy $\mathfrak{s}(\nu)$, density $\mathfrak{r}(\nu)$ and  average energy $\mathfrak{u}(\nu)$ (when a suitable interaction potential is given). 
In general, for any $\mu \in \R$ and $\beta<0$, it holds that 
\begin{equation} \label{variationalprinciple}
\begin{split}
     -\Xi(\mu, \vec 0, \beta)
     \leq
   \mathfrak{s}(\nu)   - \beta \mathfrak{u}(\nu) - \mu \mathfrak{r} (\nu),
\end{split}
\end{equation}
see \cite[\S 3]{Gii} for precise definitions and proof.
The Gibbs variational principle states that the equality in \eqref{variationalprinciple} holds if and only if $\nu$ is a Gibbs state indexed by $\mu$ and $\beta$.
If we assume that $\Xi(\mu, \vec 0, \beta)$ is not strictly convex on $\R \times \R_{<0}$ then there exists distinct 
$(\mu_1,\beta_1)$ and $(\mu_2,\beta_2)$ such that
\begin{equation} \label {eqn: strict convexity not hold}
   \Xi(\mu_t, \vec 0, \beta_t) 
    = 
   t\, \Xi(\mu_1, \vec 0, \beta_1) +(1- t)\, \Xi(\mu_2, \vec 0, \beta_2),
\end{equation}
where $(\mu_t, \beta_t) = t\,(\mu_1,\beta_1) + (1-t)\, (\mu_2,\beta_2)$ and $0<t<1$. Let $\nu$ be a Gibbs state for $(\mu_t,\beta_t)$. 
Then applying Gibbs variational principle to $\nu$,
\begin{equation}
\begin{split}
t \left [ \mathfrak{s}(\nu)   - \beta _1\mathfrak{u}(\nu) - \mu_1 \mathfrak{r} (\nu)\right] +& (1-t) \left[ \mathfrak{s}(\nu)   - \beta _2\mathfrak{u}(\nu) - \mu_2 \mathfrak{r} (\nu)\right] \\
&\quad\quad = - t \Xi(\mu_1, \vec 0, \beta_1) - (1- t) \Xi(\mu_2, \vec 0, \beta_2).
\end{split}
\end{equation} 
Again by the Gibbs variational principle, we have now that $\nu$ is also a Gibbs state of both $(\mu_1,\beta_1)$ and $(\mu_2,\beta_2)$.
However, it is standard that the Gibbs states are identifiable, in the sense that the sets of Gibbs states for different $(\mu, \beta)$ are disjoint, cf.~\cite[Remark 3.7]{Gii}.
Therefore we conclude strict convexity of $\Xi(\mu, \vec 0, \beta)$ under the assumption of Gibbs variational principle. 

The Gibbs principle holds at least for the following cases:
\begin{itemize}
\item
For superstable, finite range pairwise potentials as in \cite[Theorem 1]{D}. 
\item
For regular, non-integrably divergent pairwise potentials as in \cite[p.~1344]{Gii}. 
\end{itemize} 
Note that Lennard-Jones type interactions satisfy the assumptions from \cite{Gii} and are tempered.

The above argument was written for $\Xi(\mu, \vec 0, \beta)$ which is the case in \cite{Gii}. It is easy to see that the dependence on $\vec \l$  does not spoil strict convexity:                                                  

\begin{Prop} \label{imcludeP} If $(\mu, \beta) \mapsto \Xi(\mu, \vec 0, \beta)$ is strictly convex on $\R\times \R_{<0}$ (as in the above cases) then  
$\Xi(\aaa)$ is strictly convex on $\mathfrak{I}$.
\end{Prop}
\proof
For any $(\aaa) \in \mathfrak {I} $, by completing squares, we obtain $ \Phi_\L(\aaa) =  \Phi_\L(\mu + \tau, \vec 0, \beta) $ where
$ \tau = \tau(\vec \l, \beta) = - \dfrac {|\vec \lambda|^2}{2 \beta}$.
Therefore  $ \Xi (\aaa) =  \Xi (\mu + \tau, \vec 0, \beta) $.
Elementary calculus shows that $\tau$ is convex on $\R^3\times \R_{<0}$.
Then the strict convexity of $\Xi (\aaa)$ follows from 
the fact that $\Xi(\mu,\vec 0, \beta)$ is increasing in $\mu$ and 
the strict convexity of $(\mu, \beta) \mapsto \Xi(\mu, \vec 0, \beta)$.
\qed

\section{Convergence} \label{Main}
We are now ready for the main results. The main point here is that, thanks to the properties of the functions at hand, as the log-partition functions for finite volumes converge to the thermodynamic pressure their convex conjugates (finite volume information entropies) converge to the convex conjugate of the limit (microcanonical entropy), and so do their derivatives (the thermodynamic parameters). The section also includes comparisons with related results.

\subsection{Convergence of Epigraphs}
Recall again that the epigraph of a function $f$ is the set $\{ (x, a): f(x)\leq a \}$, that a convex function is called proper if it is not identically $+\infty$ and it never has the value $-\infty$ (with obvious modifications for concave functions), and that proper and closed is synonymous to proper and lower semi-continuous in our setting. 

For functions in general, convergence of epigraphs (for convergence of sets as in \cite[\S 4B]{RW}) is used to preserve critical points at the limit. For convex functions, pointwise convergence is not too far from convergence of epigraphs, in a sense that is made precise as we recall the following facts :

\begin{Fact}\label{pointepi}
For $f_{n}$, $f$ convex functions on $\R^{d}$ with $f_{n} \rightarrow f$, assume that $f$ is closed and the interior of the domain of $f$ is not empty. Then the epigraphs of $f_{n}$ converge to  the epigraph of $f$ as sets. 
\end{Fact}

The importance of the convergence of epigraphs also lies in that it is inherited by conjugates:

\begin{Fact}\label{convconv}
If $f_{n}$, $f$ are 
proper, closed, convex functions on $\R^{d}$, then the epigraphs of $f_n$ converge to the epigraph of $f$ if and only if the epigraphs of $f_n^*$ converge to the epigraph of $f^*$. 
\end{Fact}
Finally, convergence of epigraphs implies pointwise convergence as in the following:
\begin{Fact}\label{epipoint}
Let $f_{n}$ and $f$ be convex functions on $\R^{d}$ with the epigraphs of $f_{n}$ converging to the epigraph of $f$. Assume that $f$ is closed and the interior of the domain of $f$ is not empty. Then $f_{n} \rightarrow f$ uniformly on any compact set in the interior of the domain of $f$. 
\end{Fact}
Facts \ref{pointepi} and \ref{epipoint} are included in \cite[Theorem 7.17]{RW}, while Fact \ref{convconv} is Theorem 11.34 in the same reference. Recalling the definition of $\mathcal{E}'_{\L_l}$ from Notation \ref{notation 3}, we can now state:
\begin{Thm} \label{interiors}
For interaction $U$ stable and tempered: 
\begin{enumerate}
\item \label{interior not empty}
The interior of the domain of $s$ is not empty.
\item\label{interior is interior}
Fix any sequence of $\L$'s that approaches infinity in the strong van Hove sense and is approximable by cubes. Whenever  $(\rho, \vec u, E)$ is in the interior of the domain of $s$
then $(\rho,\vec u, E)$  is in $\mathfrak{S}_{\L}$, for all ${\L}$ large enough. 
Therefore $\displaystyle    {\rm int \,dom}(s) \subset \liminf_\L \mathfrak{S}_{\L}$.
\item\label{union interiors} 
 $\displaystyle    {\rm int \,dom}(s) = \bigcup_l {\rm int \, conv}(\mathcal{E}'_{\L_l})$.
\end{enumerate}
\end{Thm}
\proof
For \ref{interior not empty}, by \eqref{closure union} and the fact that ${\rm dom} (s)$ is convex we get
\begin{equation}
    \overline{{\rm dom}(s)}
    = 
    \overline{\bigcup_l \mathcal{E}'_{\L_l}}
    = 
    \overline{  \bigcup_l {\rm conv} \left( \mathcal{E}'_{\L_l}\right)  }.
\end{equation}
Again by convexity of ${\rm dom}(s)$, we have ${{\rm int\,dom}(s)} = {\rm int}\ {\overline{{\rm dom}(s)}}$ (cf.~\cite[Theorem 6.3]{Rock}).
It therefore follows that
\begin{equation}
 {{\rm int\,dom}(s)} 
    =
   {\rm int} \left(\overline{\bigcup_l  {{\rm conv} \left( \mathcal{E}'_{\L_l}\right)}} \right)
    \supset  \bigcup_l {\rm int}\ {{\rm conv} \left( \mathcal{E}'_{\L_l}\right)}.
 \end{equation}
As ${\rm int}\ {{\rm conv} \left( \mathcal{E}'_{\L_l}\right)}$ is not empty (cf.\,Notation \ref{notation 3}), the proof of \ref{interior not empty} is complete.

For \ref{interior is interior}: The interior of the domain of $\Xi$ is not empty, by Lemma \ref{domP} and $\Xi$ is closed since it is a conjugate as in \eqref{Xi-s}(cf.~\cite[Thm.~12.2]{Rock}). Therefore, according to  Fact \ref{pointepi}, the pointwise converge \eqref{PhitoP}  implies that the
epigraphs of $\Phi_\L$ converge to the epigraph of $\Xi$.

Therefore, by Fact \ref{convconv}, the epigraphs of $ \Phi_\L^*$ converge as sets to the epigrpaph of $\Xi^*$. Equation \eqref {-s as a conjugate} gives $ \Xi^*(\rho, \vec u, E) = -s(\rho, \vec u, E)$.
In other words, the epigraphs of $\Phi_\L^*$ converge to the epigraph of $-s$. Then apply Fact \ref{epipoint} to obtain that that for $(\rho, \vec u, E)$ in the interior of the domain of $s$  we have that $\Phi_\L^*$ converges uniformly to $-s$ on any ball in the interior of the domain that contains $(\rho, \vec u, E)$. In particular, $(\rho, \vec u, E)$ is in the interior of the domain of $\Phi_\L^*$ for all $\L$ large enough. 

For \ref{union interiors}, observe that we have already shown in part \ref{interior not empty} that
\begin{equation}
\begin{split}
    {\rm int}\ {{{\rm dom}(s)}}
    \supset 
     \bigcup_l {\rm int}\ {{\rm conv} \left( \mathcal{E}'_{\L_l}\right)} .
\end{split}
\end{equation}
Conversely, let $x \in {\rm int \ dom}(s)$. For any sequence of $\L$'s that approaches infinity in the strong van Hove sense and approximable by cubes, part \ref{interior is interior} shows that $x \in \mathfrak{S}_\L = {\rm int\ conv }( \mathcal{E}_\L)$ for all $\L$ large enough.  Now we fix the sequence of $\L$'s as $\{ \L'_l\}$. We therefore have $x \in {\rm int\ conv }( \mathcal{E}_{\L'_l})$ for $l$ large--note that here the prime is on $\L$ and not on $\mathcal{E}$. Indeed, using the convexity of the sets  $ {\rm conv }( \mathcal{E}_{\L'_l})$,  a whole neighborhood of $x$, say $B_x$, is in ${\rm conv }( \mathcal{E}_{\L'_l})$ for $l$ large enough. Using \eqref{conv prime vs without} and $\lim_{l\to \infty}\dfrac{|\L'_l|}{|\L_l|} =1$, we have that $B_x$ is in ${\rm conv }( \mathcal{E}'_{\L_l})$ for $l$ large, therefore $x \in  \bigcup_l {\rm int}\ {{\rm conv} \left( \mathcal{E}'_{\L_l}\right)} $. 
%
{See Figure \ref{Fig3} (cf.~figure in \cite[p.~102]{M-L}).}
\qed

\begin{center}
\begin{figure}
\begin{center}
\begin{tikzpicture}
[domain=0:1.,xscale=1.5,yscale=1.5]
\draw[-] (-1.2,0) -- (0,0);
\draw[domain=0:-2.]
      plot (\x, {-0.5*\x}) node[left] {$-L \rho$};
\draw[-] (0,0)-- (-1.2,1.11);    
\path[fill = gray!20!white]
(-1.8,1.9) to [out = -77, in= 145] (0,0)-- (-1.2,1.11) -- (-1.4,1.91);
\path [fill=gray!30!white] 
(0,0) 
-- (-1.2,1.11) -- (-1.4,1.91)--(0,1.91)--(0,0);
\path [fill=gray!30!white] (0,1.91)--(0,0)--(2,0)--(2,1.9)--(0,1.91);
\draw[-] (-1.2,1.11)--(-1.4,1.91);
\draw[-] (-1.2,1.11)--(0,0);
 \draw[<->] (0,2.5) node[left]{$\rho$}-- (0,0) -- (2.5,0) node[below] {$E$};
 \draw (-1.8,1.9)
 to [out = -77, in= 145] (0,0);
\end{tikzpicture}
\caption{\label{Fig3} \smaller The convex hull of $\mathcal{E}'_{\L}$ (heavier shade) compared to the the domain of $s$ (light plus heavier shade) at $\vec u = \vec 0$.}
\end{center}
\end{figure}
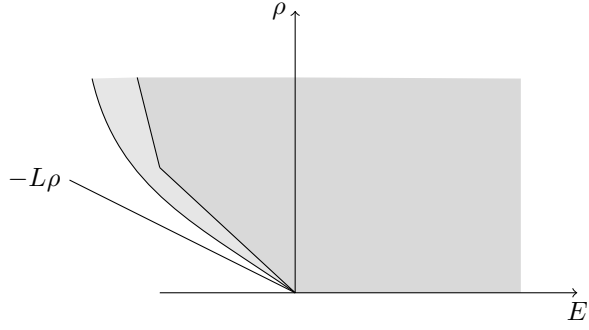
\end{center}
To prepare for the statement of the next theorem, we introduce: 
\begin{Notation} 
For $(\m)\in \R^5$, set $\tilde s_\L(\m)= - \Phi_\L^*(\m)$.
\end{Notation}  
This notation is partly justified by the relation between $\tilde s_\L$ and the information entropy $s_\L$. In fact, when $(\m)$ is in $\mathfrak{S}_\L$, by part \ref{first} of Proposition \ref{PhiStar}, we have
 $\tilde s_\L(\m) = -s_\L(g_{((\mu_\Lambda, \vec \lambda_\Lambda, \beta_\Lambda)),\L})$ for $(\mu_\Lambda, \vec \lambda_\Lambda, \beta_\Lambda)$  the unique thermodynamic parameter corresponding to $(\m)$ such that \eqref{mainEqn} holds. 
 The following theorem shows that, as $\L$ tends to infinity, $\tilde s_\L(\m)$ converges to thermodynamic limit microcanonical entropy $s(\m)$ and $(\mu_\Lambda, \vec \lambda_\Lambda, \beta_\Lambda)$ converges to $(\aaa)$, the thermodynamic parameter corresponding to $(\m)$ in the limit.
 


\begin{Thm} \label{MainTheorem}
Fix any sequence of $\L$'s that approaches infinity in the strong van Hove sense and approximable by cubes. 
For $(\m) $ in the interior of the domain of $s$, we have
\begin{enumerate}
\item 
{$\tilde s_\L(\m) \to s(\m)$, as $\L$ tends to infinity. }
\item
If $s$ is differentiable on some open convex set $C$ containing $(\m)$ (e.g.~as in section \ref{Pressure}), for $(\aaa) = -\nabla s(\m)$,
\begin{equation}
   (\mu_\L, \vec \lambda_\Lambda,
    \beta_\Lambda) 
   \to 
   (\aaa)
   \text{ as $\L$ tends to infinity}.
\end{equation}
\end{enumerate}
($(\mu_\Lambda, \vec \lambda_\Lambda, \beta_\Lambda)$ is well defined for all $\L$ large.)
\end{Thm}
\proof
{As $\tilde s_\L= - \Phi_\L^*$, then $\tilde s_\L (\m)\to -s(\m)$ follows from proof of Part \ref{interior is interior} in Theorem \ref{interiors}.}

By Part \ref{interior is interior} of Theorem \ref{interiors},
$(\m)$ is in $\mathfrak{S}_{\L}$ for all $\L$ large.Therefore for such $\L$'s, $(\mu_\Lambda, \vec \lambda_\Lambda, \beta_\Lambda)$ is well defined.
As $\Phi_\L$ is strictly convex on its domain $\mathfrak I$, $\tilde s_\L$ is differentiable on $\mathfrak{S}_{\L}$, cf.~\cite[Theorem 26.3]{Rock}. By convexity of $\mathfrak{S}_{\L}$ and Part \ref{interior is interior} of Theorem \ref{interiors}, we may find open convex set $C'$ containing $(\m)$ such that $C'\subset C$ and $C'\subset \mathfrak{S}_{\L}$ for all $\L$ large. 
Now we may use \cite[Theorem 25.7]{Rock} (roughly: for convex functions pointwise convergence implies convergence of derivatives) to conclude that
$\nabla \tilde s_\L(\m) \to \nabla s(\m)$, i.e.~
$(\mu_\L, \vec \lambda_\Lambda, \beta_\Lambda) \to (\mu, \vec \lambda, \beta)$.
\qed

\begin{Rmk}
Of less practical use is the following: 
In general, for any $(\m)\in {\rm int \ dom}(s)$, we know that the subdifferential set $\partial s$ of  $s$ at $(\m)$ is not empty, cf.~ \cite[Theorem 23.4]{Rock}. 
{Then for any $(\aaa)$ such that $-(\mu, \vec \l, \beta) \in \partial s(\rho, \vec u,E)$ there is sequence $(\rho_\L, \vec u_\L, E_\L) \to (\rho, \vec u,E)$ such that $\nabla \tilde s_\L(\rho_\L, \vec u_\L, E_\L) \to -(\mu, \vec \l, \beta)$. }
This follows from Attouche's Theorem: For $f_n$, $f$ proper, closed, convex functions, if $f_n$ converges to $f$ epigraphically  {then} $\partial f_n$  converges to $\partial f$ graphically. See \cite[\S 5E and Theorem 12.35]{RW} for definitions and proof.
\end{Rmk}

\subsection{Comparison with Maximum Likelihood Estimators} \label{MaxLike}
Several articles address the consistency of maximum likelihood estimators for Gibbs point processes, for example \cite{CG}, \cite{Gi}, \cite{DL}. The main point there is: given a Gibbs state  with parameter $\beta$, let $\omega$, a locally finite configuration on $\R^3$, be a realization of it. For $\o_\L$ the restriction of $\o$ on $\L$, maximize the likelihood
\begin{equation}
\begin{split}
      \beta_\L
      =
      {\rm argmax}_\beta 
      \frac
      {\exp\left( \beta U(\o_\L)\right) }
      {Z_\L(\beta)},
\end{split}
\end{equation}
for $U$ the interaction and $\beta<0$ to match our conventions here.
Then \cite{CG}, \cite{Gi}, \cite{DL} show that almost always, $\beta_\L \to \beta$, as $\L \to \infty$. 

For exponential families of measures, maximizing likelihood and minimizing entropy are closely related in general.
Let $x_i$, $i=1,\ldots,n$ be independent realizations of a member of an exponential family of probability measures
$
      \mu_\vartheta = 
      \exp\left( \vartheta t(x) - \log Z(\vartheta)\right)\mu_0
$.
Then maximizing the likelihood
$ \prod_i \exp\left( \vartheta t(x_i) - \log Z(\vartheta)\right)$
is the same as solving
\begin{equation}
\begin{split}
        \dfrac{d}{d\vartheta} 
        \log  \prod_i \exp\left( \vartheta t(x_i) - \log Z(\vartheta)\right) = 0
        \Leftrightarrow
        \frac1n\sum_i t(x_i)   =  \dfrac{d}{d\vartheta}\log Z(\vartheta).
\end{split}
\end{equation}
Note that the right hand side is $\mathbb{E}_\vartheta[t]$. In other words, for exponential families of the form $\exp\left(\vartheta t(x) - \log Z(\vartheta)\right)$, the maximum likelihood estimator $\vartheta$ for $\{x_i \}_{i=1}^n$ is the same as the $\vartheta$ of minimizing entropy with constraint $\sum t(x_i)/n$. See \cite[p.~82, p.~94]{Kull} and the references there for more.

It is clear then how our convergence of thermodynamics parameters differs from the consistency problem of maximum likelihood estimators: in our case, the constrains (i.e. $\rho, \vec u, E$) are given and fixed for all volumes $\L$ whereas in the consistency problems the constrains come from the restriction of the realized configuration on $\L$. Therefore, in the consistency problems the constraints not only change with $\L$, but they also depend on $\omega$, i.e.~the estimated thermodynamic parameters  are a sequence of  random variables. 

\subsection{Comparison with Lanford} \label{LanAgain}
Lanford in \cite[p.~63]{L} shows (and \cite[p.~57]{M-L} explains a crucial step in the proof) that for $\beta$ the derivative of the microcanonical thermodynamic limit entropy $s$ at energy $E$ (when this exists even in the absence of kinetic energy) for fixed density $\rho$, the information entropy of the {\em canonical} distribution with parameter $\beta$ converges to $-s(\rho,E)$. Note carefully that for Lanford $\beta$ stays the same over all volumes, whereas in our work we show convergence while the parameters change  with volume. Furthermore, the measure in Lanford's argument is canonical, rather than grand canonical. This is because $\rho$ is now fixed throughout and is not on the same footing as $E$. In particular, Lanford's $\beta$ is a function of $\rho$. 

\section{Local Homeomorphism at the Thermodynamic Limit}\label{localbijection}
We show here that at the thermodynamic limit the bijection between thermodynamic parameters and macroscopic quantities as in Theorem \ref{bijection}
holds at least locally in some region. The existence of such a bijection  for continuous systems is folklore in the theory of hydrodynamic limits, see for example \cite[p.~530, p.~556]{OVY}. For the case of particle configurations on the line see \cite[Proposition 5.3]{LR}, \cite[Theorem 10.2]{V}. 

Throughout this section, in addition to stability and temperedness, we will assume pair interactions:  $U(q_1,\ldots,q_n) = \sum_{i \neq j} \Phi(q_i - q_j)$. 
Notice that the stability \eqref{stability} and temperedness  as in \eqref{condition: temperedness} imply that 
\begin{equation} \label{Ruelle}
\begin{split}
     C(\beta)
     :=
     \int\limits_{\R^3} 
     \left|\exp\left(\beta \Phi(x)\right) -1\right| 
     dx< +\infty, 
     \quad \text{for all $\beta<0$},
\end{split}
\end{equation}
cf.\,\cite[p.~32, p.~72]{R}.
To establish the  local homeomorphism, we will need the thermodynamic limit pressure $\Xi$ to be both strictly convex 
(as in section \ref{section: strict convexity}) and differentiable. It is standard that, when not taking velocities into account,
 i.e.~$ \mathfrak X_\L = \bigsqcup_{n\geq 0} \L^n$, $H=U$ and $\vec \l = 0$,
$\Xi$ is analytic in $\mu$ and $\beta$ in the low density region 
\begin{equation} 
\begin{split}
    \left
    \{(\mu, \beta): 
    \beta< 0, 
    \mu < 2\beta L -1 - \log C(\beta)
    \right\}.
\end{split}
\end{equation}
Analyticity with respect to $\mu$ is shown in \cite[Theorem 4.3.1]{R}. Combining this with results from \cite{LP}, analyticity with respect to $\beta$ also follows. This is presented in detail in \cite[Appendix D]{X}.
 
Including the kinetic energy and $\vec \l$, it follows easily that we have differentiability of $\Xi(\mu, \vec \l, \beta)$ for
\begin{equation}
\begin{split}
     \mathcal{R}
    =
    \left\{
    (\mu, \vec \l, \beta): 
    \beta < 0, 
     \mu 
     <
     2\beta L -1 
     - \log C(\beta) 
     - \log 
     \int\limits_{\R^3} 
     \exp\left(\vec \l \cdot \vec p + \beta\dfrac{|\ \vec p\ |^2}2\right)d\vec p
     \right\}.
\end{split}
\end{equation}

\begin{Prop}
Let $U$ be a stable, tempered, pair interaction potential with $\Xi$ strictly convex (e.g.~as in section \ref{section: strict convexity}) and let $K$ be convex, open subset of $\mathcal R$. Then
$\nabla \Xi:K \to \R^5$ is a homeomorphism onto $\nabla \Xi(K)$ with $(\nabla \Xi)^{-1} = -\nabla s$.
\end{Prop}

\proof  
On such a $K$ the pressure $\Xi$ is both differentiable and strictly convex which implies that $\nabla \Xi$ is one-to-one from $K$ to $\nabla \Xi (K)$, using \cite[Theorem B, p.~99]{RV} again. $\nabla \Xi$ is also continuous, see \cite[Theorem 25.5]{Rock}.

$\nabla \Xi(\mu, \vec \l, \beta) = (\rho, \vec u, E)$ if and only if $(\mu, \vec \l, \beta)$ is in the subdifferential of $-s$ at $(\m)$ (\cite[Theorem 23.5]{Rock}).
Therefore, as $-s$ is essentially differentiable by Lemma \ref{s essentially differentiable}, the range $\nabla \Xi (K)$ can only be a subset of the interior of the domain of $-s$, where $-s$ is differentiable, and $\nabla \Xi = (\nabla(-s))^{-1}$. Since $\nabla s$ is also continuous, $\nabla \Xi$ is a homeomorphism.
\qed


\begin{thebibliography}{Gibbs}
\bibitem[BAR]{BAR}
Balian R., Alhassid, Y., Reinhardt, H.: 
{Dissipation in many-body systems: a geometric approach based on information theory.} Phys. Rept. {\bf 131}, 1--46 (1986)

\bibitem[BKM]{BKM} Barlow H.B.; Kaushal, T.P., Mitchison,G.J.: {Finding Minimum Entropy Codes.} Neural Computation {\bf 1}, 412--423 (1989)

\bibitem[BCM]{BCM}
Bot, R.I., Csetnek, E.R., Moldovan, A.
Revisiting Some Duality Theorems via the Quasirelative Interior in Convex Optimization, J. Optim. Theory Appl. {\bf 139}, 67--84 (2008).

\bibitem[CG]{CG} 
Comets, F., Gidas, B.:
{Parameter estimation for Gibbs distributions from partially observed data}. The Annals of Applied Probability {\bf 2} (1), 142--70 (1992)

\bibitem[D]{D}
Dereudre, D.:
{Variational Principle for Gibbs point processes with finite range interaction}.
Electron. Commun. Probab. {\bf 21} (10), 1--11 (2016)

\bibitem[DL]{DL} Dereudre D., Lavancier, F.: {Consistency of likelihood estimation for Gibbs point processes}. The Annals of Statistics {\bf 45} (2), 744--770 (2017)

\bibitem[G]{Gii} Georgii, H.-O.: {The Equivalence of Ensembles for Classical Systems of Particles.} Journal of Statistical Physics {\bf 80}, (5/6), 1341--1378 (1995)

\bibitem[Gibbs]{Gibbs}
Gibbs, J.W.: {Elementary Principles In Statistical Mechanics}.
Ox Bow Press, Connecticut (1981)

\bibitem[Gi]{Gi} Gidas, B.: {Parameter Estimation for Gibbs Distributions from Fully Observed Data.} In: Chellappa, R., Jain, A. (eds.) Markov Random Fields, Theory and Application, pp.~471--498. Academic Press (1993)

\bibitem[GR]{GR}
Griffiths R.B., Ruelle, D.:
{Strict convexity (``continuity") of the pressure in lattice systems.} Communications in Mathematical Physics {\bf 23} (3),169--75 (1971)

\bibitem[J]{J} Jaynes, E.T.: {Information Theory and Statistical Mechanics.}
In: Statistical Physics. Brandeis Lectures, vol.~3. Benjamin (1963).

\bibitem[Kh]{Kh} Khinchin, A.I.:
{Mathematical Foundations of Statistical Mechanics.}
Dover (1949)

\bibitem[K]{K} Kullback S.: 
{Certain inequalities in information theory and the Cramer-Rao inequality.} The Annals of Mathematical Statistics {\bf 25} (4), 745--751 (1954).

\bibitem[Kull]{Kull}
Kullback S.: {Information Theory and Statistics.} Dover (1968)

\bibitem[L]{L} 
Lanford, O.E.: {Entropy and equilibrium states in classical statistical mechanics.} In: Statistical mechanics and mathematical problems, pp.~1--113. Springer (1973)

\bibitem[LR]{LR}
Lanford O.E., Ruedin, L.:
{Statistical mechanical methods and continued fractions.}
Helvetica Physica Acta {\bf 69}(5), 908--948 (1996)

\bibitem[LP]{LP}
Lebowitz, J. L.; Penrose, O.: 
{Analytic and clustering properties of thermodynamic functions and distribution functions for classical lattice and continuum systems.} Comm. Math. Phys. {\bf 11}, 99--124 (1968/1969)

\bibitem[M-L]{M-L} Martin-L\"of, A.:
{Statistical mechanics and the foundations of thermodynamics.}
Springer (1979)

\bibitem[OVY]{OVY}Olla S., Varadhan S.R., Yau H.T.:
{Hydrodynamical limit for a Hamiltonian system with weak noise.} Communications in mathematical physics. {\bf 155} (3) 523--60 (1993).

\bibitem[PS]{PS}
Pistone G, Sempi, C.: 
{An infinite-dimensional geometric structure on the space of all the probability measures equivalent to a given one.}
 The Annals of Statistics {\bf 23} (5), 1543-61 (1995).

\bibitem[P]{Preston}Preston, C.:  
{Random fields.}
 Lecture Notes in Math.~Vol. 534. Springer (1976).

\bibitem[RP]{RP} 
Rechtman, R., Penrose, O.:
{Continuity of the Temperature and Derivation of the Gibbs Canonical Distribution in Classical Statistical Mechanics.}
Journal of Statistical Physics, {\bf 19} (4), 359--366 (1978).

\bibitem[RV]{RV} Roberts, A.W., Varberg, D.E.: 
{Convex Functions.} Academic Press (1973)

\bibitem[Rock]{Rock}
Rockafellar R.T.: {Convex analysis.} Princeton University press (1970)

\bibitem[RW]{RW} Rockafellar R.T., Wets, R.: {Variational Analysis.} Springer (1998)

\bibitem[R]{R} 
Ruelle, D.: 
{Statistical Mechanics; Rigorous results.}
Benjamin (1969)

\bibitem[V]{V} Varadhan, S.R.S.: {Scaling Limits for Interaction Diffusions.} Comm. Math. Phys. {\bf 135}, 313--353 (1991)


\bibitem[WJ]{WJ}Wainwright, M.J., Jordan, M.I.:
{Graphical Models, Exponential Families, and
Variational Inference.}
Foundations and Trends in Machine Learning {\bf 1} (1--2), 1--305,  (2008)

\bibitem[X]{X} Jianfei Xue: 
{On Hydrodynamic Equations and Their Relation to Kinetic Theory and Statistical Mechanics}. University of Missouri Ph.D.~Thesis (2017) 

\bibitem[Zei]{Zei}
Zeidler, E. 
Nonlinear Functional Analysis and its Applications III,
Springer (1985)

\bibitem[Z]{Z} Zubarev. D.N., Morozov, V., R\"opke, G.: 
{Statistical Mechanics of Nonequilibrium Processes, Vol. 1.}
Akademie Verlag (1996)

\end{thebibliography}
\end{document}